\documentclass[12pt, draftclsnofoot, onecolumn]{IEEEtran}
\usepackage{amsfonts}
\usepackage{dsfont}
\usepackage{amssymb}
\usepackage{graphicx}
\usepackage{subfigure}
\usepackage{enumerate}
\usepackage{amsmath}
\usepackage{color}
\usepackage{amsthm}
\usepackage{amsmath}
\usepackage{tabularx}
\usepackage{algorithm}
\usepackage{algpseudocode}
\usepackage[bookmarks=false]{hyperref}
\hypersetup{hidelinks}
\usepackage{breakurl}
\usepackage{cite}
\usepackage{epstopdf}
\usepackage{comment}
\usepackage{multirow}
\allowdisplaybreaks[4]

\DeclareMathOperator*{\argmin}{arg\,min}
\DeclareMathOperator*{\argmax}{arg\,max}

\newcommand{\bp}{\begin{proof} \small }
\newcommand{\ep}{\end{proof} \normalsize}
\newcommand{\epx}{\end{proof} \small}
\newcommand{\bpa}{\begin{proofappx} \footnotesize }
\newcommand{\epa}{\end{proofappx} \small }
\newtheorem{theorem}{Theorem}

\newtheorem{corollary}{Corollary}
\newtheorem{lemma}{Lemma}

\newtheorem*{theorem*}{Theorem}
\newtheorem*{proposition*}{Proposition}
\newtheorem*{corollary*}{Corollary}
\newtheorem*{lemma*}{Lemma}
\newtheorem*{assumption*}{Assumption}
\newtheorem*{definition*}{Definition}
\newtheorem*{claim*}{Claim}
\newtheorem{remark}{Remark}

\newcommand{\be}{\begin{equation}}
\newcommand{\ee}{\end{equation}}
\newcommand{\bs}{\begin{subequations}}
\newcommand{\es}{\end{subequations}}
\newcommand{\bq}{\begin{eqnarray}}
\newcommand{\eq}{\end{eqnarray}}
\newcommand{\bqn}{\begin{eqnarray*}}
\newcommand{\eqn}{\end{eqnarray*}}

\newcommand{\ba}{\left[ \begin{array}}
\newcommand{\ea}{\\ \end{array} \right]}
\newcommand{\ben}{\begin{enumerate}}
\newcommand{\een}{\end{enumerate}}

\def\d{{\boldsymbol{d}}}

\def\real{{\mathchoice%
{\hbox{\rm\setbox1=\hbox{I}\copy1\kern-.45\wd1 R}}
{\hbox{\rm\setbox1=\hbox{I}\copy1\kern-.45\wd1 R}}
{\hbox{\scriptsize\rm\setbox1=\hbox{I}\copy1\kern-.45\wd1 R}}
{\hbox{\scriptsize\rm\setbox1=\hbox{I}\copy1\kern-.45\wd1 R}}}}

\def\Zint{{\mathchoice{\setbox1=\hbox{\sf Z}\copy1\kern-.75\wd1\box1}
{\setbox1=\hbox{\sf Z}\copy1\kern-.75\wd1\box1}
{\setbox1=\hbox{\scriptsize\sf Z}\copy1\kern-.75\wd1\box1}
{\setbox1=\hbox{\scriptsize\sf Z}\copy1\kern-.75\wd1\box1}}}
\newcommand{\complex}{ \hbox{\rm C\kern-0.45em\rule[.07em]{.02em}{.58em}%
\kern 0.43em}}

\begin{document}
	%\title{Task Replication for Vehicular Edge Computing}
	%\title{Task Replication for Vehicular Edge Computing: A Combinatorial Multi-Armed Bandit based Approach}
	\title{Distributed Task Replication for Vehicular Edge Computing: Performance Analysis and Learning-based Algorithm}
	
	\author{Yuxuan~Sun,~\IEEEmembership{Student Member,~IEEE,}
		Sheng~Zhou,~\IEEEmembership{Member,~IEEE,} \\
		Zhisheng Niu,~\IEEEmembership{Fellow,~IEEE} 
		\thanks{Y. Sun, S. Zhou and Z. Niu are with the Department of Electronic Engineering, Tsinghua University, Beijing 100084, China. Emails: sunyx15@mails.tsinghua.edu.cn, \{sheng.zhou, niuzhs\}@tsinghua.edu.cn.}  
		\thanks{This work is sponsored in part by the Nature Science Foundation of China (No. 61871254, No. 91638204, No. 61861136003), National Key R\&D Program of China 2018YFB0105005, and Intel Collaborative Research Institute for Intelligent and Automated Connected Vehicles. 
			(Corresponding author: Sheng Zhou.)}
		\thanks{Part of this work has been presented in IEEE GLOBECOM 2018 \cite{Sun2018globecom}.}
	}
	
	\maketitle

	\begin{abstract}
		In a vehicular edge computing (VEC) system, vehicles can share their surplus computation resources to provide cloud computing services. The highly dynamic environment of the vehicular network makes it challenging to guarantee the task offloading delay. To this end, we introduce task replication to the VEC system, where the replicas of a task are offloaded to multiple vehicles at the same time, and the task is completed upon the first response among replicas. First, the impact of the number of task replicas on the offloading delay is characterized, and the optimal number of task replicas is approximated in closed-form. Based on the analytical result, we design a learning-based task replication algorithm (LTRA) with combinatorial multi-armed bandit theory, which works in a distributed manner and can automatically adapt itself to the dynamics of the VEC system. A realistic traffic scenario is used to evaluate the delay performance of the proposed algorithm. Results show that, under our simulation settings, LTRA with an optimized number of task replicas can reduce the average offloading delay by over $30\%$ compared to the benchmark without task replication, and at the same time can improve the task completion ratio from $97\%$ to $99.6\%$.\\
		
	\end{abstract}
	
	\begin{IEEEkeywords}
		Vehicular edge computing, computation task offloading, task replication, online learning, combinatorial multi-armed bandit.
	\end{IEEEkeywords}	
	
	% \fi
	\IEEEpeerreviewmaketitle
	
	\section{Introduction}
	
	Vehicles are becoming \emph{connected} and \emph{intelligent}. The development of communications protocols such as IEEE 802.11p and LTE-V enable vehicle-to-everything (V2X) communications \cite{80211p,3gpp36300,tdlte}. To realize autonomous driving and various on-board infotainment applications, vehicles will be equipped with powerful computation resources, e.g., to handle $10^6$ dhrystone million instructions per second \cite{intel}, as well as a variety of sensors such as cameras and radars. 
	These \emph{moving} communication, computation and sensing resources can be further exploited to enhance conventional multi-access edge computing (MEC) systems \cite{mao2017mobile, mach2017mobile, yu2018survey}, where computation and storage resources are deployed in static infrastructures such as base stations (BSs) at the edge of wireless networks.
	
	Consequently, the concept of \emph{vehicular edge computing} (VEC) (also known as vehicular fog or cloud computing) has been proposed \cite{abdel2015vehicle,bitam2015vanet,hou2016vfc,choo2017sdvc,qian2018collaborative,Sun2019CMag, Abdreev2019dense,cheng20175G}, where road side units (RSUs) and vehicles with surplus computation resources are employed as computing nodes just like the role of edge servers in the MEC system. The computation resources are abstracted via network function virtualization and software defined networking techniques to support various applications.
	Task requesters, including on-board driving systems and mobile devices of passengers and pedestrians, can get computing services from service providers, including vehicles and RSUs, by means of \emph{task offloading}. 
	In this context, vehicles acting as service providers are called \emph{service vehicles} (SeVs), while vehicles whose driving systems or passengers requesting computation task offloading are called \emph{task vehicles} (TaVs).
	Typical use cases in the VEC system include autonomous driving applications such as collective environment perception and cooperative collision avoidance \cite{5gv2x}, and vehicular crowd-sensing for road monitoring and parking navigation \cite{vehcrowd}. Applications in conventional MEC systems are also supported by the VEC system for passengers and pedestrians, including augmented reality, cloud gaming, and etc.

	%For safety-related applications and other infotainment applications i
	In the VEC system, the \emph{offloading delay}, including data transmission and computation, is the key performance metric, and it is vital to schedule tasks and allocate computation resources for real-time computing services.
	Tasks can be offloaded from TaVs to SeVs directly in a \emph{distributed} manner, or collected by the RSUs and then assigned to the SeVs in a \emph{centralized} manner \cite{Sun2019CMag}.
	In the literature, centralized resource allocation schemes are proposed in \cite{choo2017sdvc,zheng2015smdp,Jiang2017IoT}, wherein the communication and computation resources are optimized globally based on Markov decision process (MDP) with the coordination of RSUs. However, the complexity is usually very high due to a large state space involving many vehicles and tasks. The global states, including locations, velocities, moving directions of vehicles, wireless channel states and available computation resources, should also be collected by RSUs frequently, leading to high signaling overhead.
	An alternative way is to make offloading decisions in a distributed manner by task requesters \cite{fengave,zhou2019reliable , Sun2018ICC,Sun2019TVT}. 
	In this context, it is still difficult for the TaV to acquire the global state information of SeVs and the offloading behaviors of other TaVs in the neighborhood. Contract theory is adopted in \cite{zhou2019reliable}, while online learning algorithm based on multi-armed bandit (MAB) theory is proposed in \cite{Sun2018ICC, Sun2019TVT}, to overcome the challenges.
	
	In fact, challenges and opportunities coexist in the VEC system.
	On the one hand, task offloading in the VEC system faces a more volatile environment, where the network topology and wireless channels vary rapidly due to vehicle movements.
	On the other hand, moving vehicles acting as VEC servers can provide more computation offloading opportunities, while at the same time relieving the impact from the voltile environment.
	
	To further exploit the computation resources in the VEC system, we introduce \emph{task replication}. Specifically, each task is replicated to multiple candidate SeVs at the same time and executed by them independently. Upon the first result transmitted back from one of the selected SeVs, the task is completed.
	Task replication technique is adopted in large-scale cloud computing servers to reduce delay and mitigate the straggler effect, and the impact of redundancy level (i.e., the number of task replicas) on the delay performance is analyzed based on queueing theory \cite{gardner2015reducing,gardner2017redundancy,joshi2017efficient}.
	However, the transmission delay and the dynamic network topology are not considered in these works.
	Introducing task replication to the VEC system, a centralized algorithm that maximizes the task completion ratio is proposed in \cite{Jiang2017IoT} based on MDP, while a contextual MAB based learning algorithm is proposed in \cite{chen2019task}, enabling the RSU to treat the service delay as a grey box. However, no theoretical results have been revealed on how to select the number of task replicas under different network conditions, including density of vehicles, task arrival rates and service capabilities, to optimize the quality of service such as delay and reliability.
	Moreover, these algorithms require the RSUs to collect and assign tasks in a centralized manner.
	
	In this paper, we study the task replication problem in the VEC system, with the objective of delay minimization. 
	The optimal number of task replicas is derived, and a distributed learning-based task replication algorithm is proposed.
	% derive the optimal number of task replicas of the VEC system and design a distributed task replication algorithm to minimize the average offloading delay. 
	The main contributions are summarized as follows:
	\begin{itemize}
		\item We propose a distributed task replication framework, which enables any TaV to offload task replicas to multiple candidate SeVs in a distributed manner, so as to minimize the average offloading delay under the task failure constraint. 
		
		\item Performance analysis is carried out to characterize how the number of task replicas affects the offloading delay, given the network conditions such as density of vehicles, average task arrival rate and computing power. The optimal number of task replicas is approximated in closed-form, and is validated through simulations.
		
		\item Exploiting combinatorial MAB (CMAB) theory, an online learning-based task replication algorithm (LTRA) is proposed, which can adapt to the dynamics of the VEC system, with provable bounded learning regret.
		
		\item A realistic traffic scenario is generated via traffic simulator Simulation for Urban MObility (SUMO) to evaluate the proposed task replication algorithm. Results show the delay reduction brought by the joint effort of task replication and online learning. Specifically, under our settings, using LTRA with the optimal number of task replicas can reduce the average offloading delay by over $30\%$, compared to the benchmark without task replication. Meanwhile, the task completion ratio can be improved from $97\%$ to over $99.6\%$.
	\end{itemize}
	
	The rest of this paper is organized as follows.
	In Section \ref{mod}, we introduce the system model and problem formulation.
	The impacts of the number of task replicas on the delay performance and task failure probability are analyzed in Section \ref{sys_ana}, and the task replication algorithm is then proposed in Section \ref{algo}.
	Numerical and simulation results are shown in Section \ref{sim}, and the paper is finally concluded in Section \ref{con}. 
	
	Throughout the paper, we use $\mathbb{E}(\cdot)$ to represent the expectation operation, and $\mathbb{P}[\cdot]$ to represent the probability of an event.
	%and $\text{Exp}(\mu)$ to represent an exponential distribution with rate parameter $\mu$. 
	Define $\mathbb{I}\{x\}$ as an indicator function, where $\mathbb{I}\{x\}=1$ if condition $x$ is true, and $\mathbb{I}\{x\}=0$ otherwise.
	The cardinality of a set is denoted by $|\cdot|$. Let $\binom{n}{k}$ denote the number of combinations of choosing $k$ items out of $n$ at a time.
	$\left\lceil x \right\rceil$ maps the least integer equal to or greater than $x$, and $\text{round}(x)$ maps $x$ to its nearest integer.
	\begin{figure*}[!t]
		\centering	
		\includegraphics[width=1\textwidth]{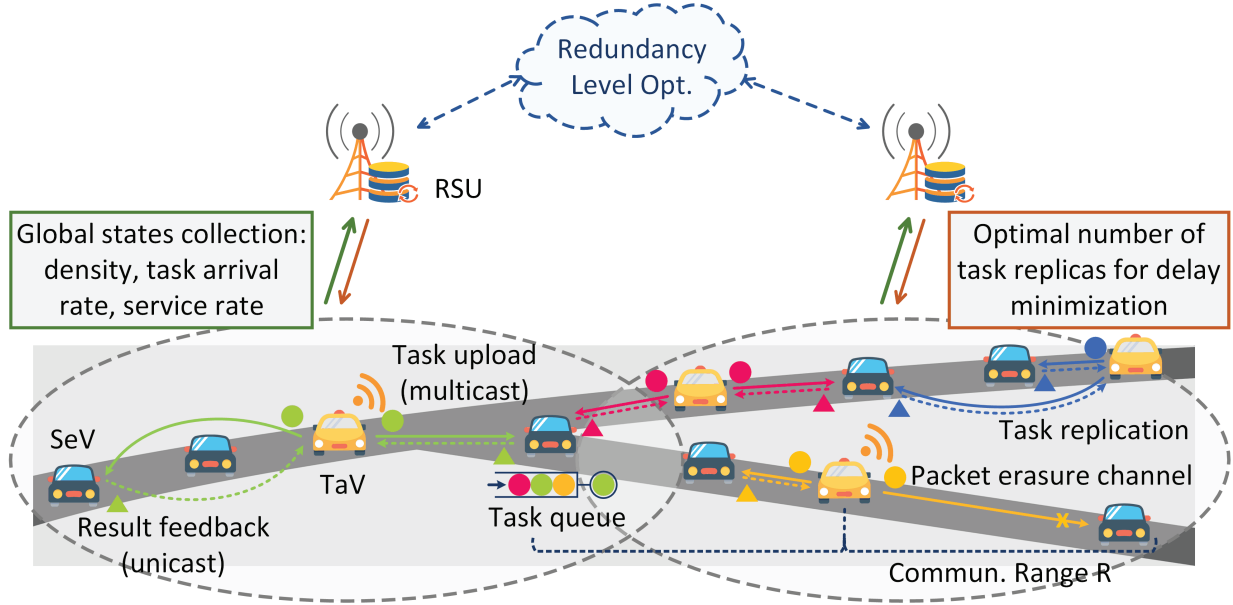}	
		\caption{Illustration of task replication in the VEC system. The RSUs collect general network conditions, optimize the number of task replicas based on the global states for delay minimization, and broadcast the optimal number of task replicas to the vehicles. Meanwhile, TaVs offload the replicas of tasks to the neighboring SeVs in a distributed manner, which involves three procedures: task upload, task execution and result feedback.}
		\label{sys}
	\end{figure*} 
	
	\section{System Model and Problem Formulation} \label{mod}
	%\begin{table}[!htb]
	%	\caption{Summary of Notations}
	%	\centering
	%	\begin{tabular}{p{1.5cm}<{\centering}|p{5cm}<{\centering}}
	%		\hline
	%		\hline
	%		\textbf{Notation} & \textbf{Description} \\
	%		\hline
	%	    $R$ & Communication range.   \\
	%	    	\hline
	%	    $\mathcal{N}_t$ & Candidate SeV set for task $t$.  \\
	%	    	\hline
	%	    $\mathcal{S}_t$ &  Selected SeV set for task $t$. \\
	%	    	\hline
	%	    $\lambda_0$ & Task arrival rate at TaV.   \\
	%	    	\hline
	%	    &   \\
	%	    	\hline
	%	    &   \\
	%		\hline
	%		\hline
	%	\end{tabular}
	%\end{table}
	
	\subsection{System Overview}
	
	As shown in Fig. \ref{sys}, we consider computation task offloading among moving vehicles in a VEC system. \emph{TaVs} generate computation tasks and offload them to the neighboring \emph{SeVs}, with surplus computation resources, for remote execution. Note that the role of each vehicle, i.e., being either a TaV or an SeV, can change across time, which mainly depends on whether it has free computation resources to share.
	
	For each TaV, the SeVs moving in the same direction within its single-hop communication range $R$ are considered as candidates. Multiple candidate SeVs may be able to process each task, and we introduce \emph{task replication} technique to counter the dynamics in the vehicular networks. In particular, each task is replicated and offloaded to multiple candidate SeVs and processed by them independently. Once the first response is received from one of the selected SeVs, the task is completed.
	However, we assume that other slower SeVs do not cancel the replicas of this task upon its completion, due to two main reasons.
	On the one hand, cancellations require TaVs to exchange additional information with SeVs, and cancelling tasks at SeVs introduces additional delay \cite{gardner2017better}, which further complicate the task replication problem.
	On the other hand, the proposed reinforcement learning algorithm needs to observe the delay performance of all the selected SeVs, which will be introduced in Section \ref{algo}.
	
	% to reduce the signaling overhead and enable the reinforcement learning algorithm, which will be introduced in Section \ref{algo}.
	%[An example is given in Fig. [XXX], where ...]
	
	Distributed task offloading is considered in this work. That is, each TaV selects a subset of candidate SeVs to serve each task on its own, without any coordinations with other vehicles.
	Without loss of generality, we will focus on a typical TaV and design the task replication algorithm in the following.
	
	\subsection{SeV Candidates and Task Arrival}
	In the considered time period, the TaV has a total number of $T$ tasks to be offloaded. For the $t$-th task, the candidate SeV set is denoted by $\mathcal{N}_t$, which may vary across time due to vehicle movements. These candidate SeVs may also serve the tasks from other TaVs simultaneously. We assume that the TaV can always associate to at least one SeV during the considered period, i.e., $\mathcal{N}_t \neq \emptyset$ for $\forall t=1,\cdots,T$. Otherwise, the TaV may offload tasks to RSUs, which is beyond the scope of this paper. A subset of SeVs, denoted by $\mathcal{S}_t\subseteq\mathcal{N}_t$, are selected to process the replicas of the $t$-th task. We assume that the number of task replicas is fixed as $K$, where $K$ should be carefully designed based on the network conditions, such as the densities of TaVs and SeVs, task arrival rates at TaVs, service capabilities of SeVs, etc. If $ |\mathcal{N}_t| \geq K$, we have $|\mathcal{S}_t|=K$ and $\mathcal{S}_t \subseteq \mathcal{N}_t$. If $ |\mathcal{N}_t|<K$, let $\mathcal{S}_t =\mathcal{N}_t$.
	
	The arrival of tasks at the TaV is modeled by a Poisson process with rate $\lambda_0$ (in tasks per second). The input data size to be uploaded from TaV to SeV, and the output data size transmitted back from SeV to TaV, are assumed to be identical across time, and denoted by $L_i$ and $L_o$ (in bits) respectively. This is a reasonable assumption since tasks generated from the same kind of applications have similar input and output data size. Moreover, tasks with different input data volumes can be further partitioned into subtasks and offloaded sequentially. For example, video segmentation technique can be adopted to partition long video frames into short video clips for classification or detection purposes \cite{cvpr2010}.
	
	\subsection{Task Replication Procedure}
	In the VEC system, there are three procedures for task replication, i.e., task upload, task execution and result feedback.
	Details of the three procedures and the corresponding delay models are described in the following.
	
	\subsubsection{Task Upload}
	%\emph{Task Upload:} 
	We consider that the replicas of each task are multicast to all the selected SeVs without retransmission using IEEE 802.11p protocol \cite{80211p}, so that replication \emph{does not} bring additional communication burdens to the VEC system. For the $t$-th task, the uplink wireless channel state between the TaV and SeV $n \in \mathcal{N}_t$ is denoted by $h_u(t,n)$, and the interference power is denoted by $I_u(t,n)$. Given the uplink channel bandwidth $W_u$, transmission power $P$ and noise power $N_0$, the achievable uplink transmission rate between TaV and SeV $n \in \mathcal{N}_t$, denoted by $r_u(t,n)$, is given by
	\begin{align} \label{uplink_raten}
	r_u(t,n) = W_u\log_2\left(1 + \frac{Ph_u(t,n)}{N_0 + I_u(t,n)}\right).
	\end{align}
	
	We assume that the transmission link between TaV and SeV $n$ is a packet erasure channel with erasure probability $p_e(t,n)$. That is, the input data of the $t$-th task is either successfully received by SeV $n$ with probability $1-p_e(t,n)$, or failed with probability $p_e(t,n)$, due to the randomness of channels such as blockage or deep fast fading. We also assume that the erasure behavior of each SeV is independent from other SeVs.
	
	Note that all the selected SeVs should be able to receive the task via multicast. Therefore, given the selected SeV set $\mathcal{S}_t$, the achievable uplink transmission rate is given by
	\begin{align} \label{uplink_rate}
	r_u(t,\mathcal{S}_t) = \min_{n\in \mathcal{S}_t} r_u(t,n) .
	\end{align}
	Then the uploading delay, denoted by $d_u(t,\mathcal{S}_t)$, can be written as
	\begin{align}  
	d_u(t,\mathcal{S}_t) = \frac{L_i }{r_u(t,\mathcal{S}_t)}.
	\end{align}
	
	\subsubsection{Task Execution}
	%\emph{Task Execution:} 
	For the $t$-th task, the SeVs that can successfully receive its input data is denoted by $\mathcal{S}'_t$, with $\mathcal{S}'_t\subseteq \mathcal{S}_t$. A task failure occurs when all the selected SeVs fail to receive the input data packets, with probability
	\begin{align}
	p_f(t,\mathcal{S}_t)=\prod_{n\in \mathcal{S}_t} p_e(t,n).
	\end{align}
	
	%Task arrival at each candidate SeV $n$ is approximated by a Poisson process with rate $\lambda_c(t,n)$ (in tasks per second), based on the fact that tasks offloaded from the TaVs it serves are independent, and the superposition of independent arrivals resembles a Poisson process \cite{Sriram1983characterizing,Ko2018wireless}. 
	
	Each candidate SeV $n$ can serve multiple TaVs within its service range $R$, the the offloading behavior of these TaVs are independent from others. Based on the fact that the superposition of independent arrivals resembles a Poisson process  \cite{Sriram1983characterizing,Ko2018wireless}, the task arrival at each candidate SeV $n$ is approximated by a Poisson process with rate $\lambda_c(t,n)$ (in tasks per second).
	Parameter $\lambda_c(t,n)$ reflects the workload of each SeV $n$, which is related to the number of TaVs within its communication range, as well as the task arrival rate and the number of replicas of each TaV.
	
	Task execution at each SeV is modeled by an M/M/1 queueing system according to the first-come first-served discipline, where the service rate of SeV $n$ is denoted by $\mu_c(n)$ (in tasks per second), and the task processing delay (service time) follows exponential distribution with mean $\frac{1}{\mu_c(n)}$.
	Define the total task execution delay (sojourn time) of the $t$-th task as $d_c(t,n)$, which includes queueing delay and processing delay. According to queueing theory \cite{queueing_theo}, task execution delay $d_c(t,n)$ follows exponential distribution with mean $\frac{1}{\mu_c(n)-\lambda_c(t,n)}$.
	
	\subsubsection{Result Feedback}
	%\emph{Result Feedback:} 
	Upon completion, each selected SeV $n \in \mathcal{S}'_t$ unicasts the computation results back to the TaV using a spectrum orthogonal to that for task uploading. 
	We assume that the results can always be delivered back to the TaV successfully, with retransmissions if necessary. Define the result feedback delay as $d_d(t,n)$, including queueing delay and transmission delay, which may be affected by many factors such as downlink channel state, interference power, link reliability and retransmission times. 
	The expression of $d_d(t,n)$ is not specified in our work, since the proposed task replication algorithm in Section \ref{algo} can learn it.
	%We do not specify the expression of $d_d(t,n)$ in our work, and the proposed task replication algorithm can handle the case without explicit expressions. 
	
	\begin{comment}
	Considering that the output data size $L_o$ is usually small, we assume that the result feedback delay is much shorter than the task processing, and thus there is no queues before downlink transmission. Define the downlink channel state and the interference power as $h_d(t,n)$ and $I_d(t,n)$, and $W_d$ the downlink channel bandwidth. Downlink transmission rate is given by
	\begin{align} 
	r_d(t,n) = W_d\log_2\left(1 + \frac{Ph_d(t,n)}{N_0 + I_d(t,n)}\right).
	\end{align}
	Thus the result feedback delay from each SeV $n \in \mathcal{S}'_t$ to the TaV is
	\begin{align}  
	d_d(t,n) = \frac{L_o }{r_d(t,n)}.
	\end{align}
	\end{comment}
	
	\subsection{Problem Formulation}
	Conditioned on $\mathcal{S}'_t \neq \emptyset$, i.e., at least one SeV successfully receives the input data of the $t$-th task, the \emph{offloading delay} $d(t,\mathcal{S}_t)$, including task upload, execution and result feedback delay, can be written as
	\begin{align}
	d(t,\mathcal{S}_t)=d_u(t,\mathcal{S}_t)+\min_{n \in \mathcal{S}'_t} \left( d_c(t,n) +d_d(t,n)   \right).
	\end{align}
	
	The objective is to minimize the average offloading delay of $T$ tasks under a failure probability constraint, by optimizing the task replication decisions $\mathcal{S}_t$:
	\begin{subequations}
		\begin{align}
		\textbf{P1:}~\min_{\mathcal{S}_1,...,\mathcal{S}_T}& \frac{1}{T}\sum_{t=1}^{T} \mathbb{I}\{\mathcal{S}'_t \neq\emptyset\}d(t,\mathcal{S}_t) \\
		\text{s.t.}~~~ &\frac{1}{T}\sum_{t=1}^{T} p_f(t,\mathcal{S}_t) \leq \theta_f, \\
		&\mathcal{S}_t\subseteq\mathcal{N}_t, ~t=1,\cdots,T,
		\end{align}
	\end{subequations}
	where $\theta_f$ is the threshold of the task failure probability.
	
	In practical VEC systems, it is impossible for each TaV to acquire the \emph{future} state information, including future candidate SeVs and the corresponding transmission rates, packet erasure probabilities, etc. In addition, due to the limited signaling resources and the decentralized nature of VEC systems, it is also very difficult for the typical TaV to acquire the current \emph{global} state information such as the densities of TaVs in the neighborhood, the workloads of candidate SeVs and the wireless channel states. 
	Consequently, the TaV has no idea how to make task replication decisions to solve problem \textbf{P1}, i.e., how many SeVs and which SeVs to select.
	%As a result, the TaV cannot decide how to offload the tasks.
	
	To deal with the aforementioned challenges, we will solve the problem in two steps in the following. First, we analyze the optimal number of task replicas from the centralized view, based on the general network conditions collected by RSUs. Based on this result, we further propose a \emph{learning while offloading} solution to enable the TaV to learn the delay performance of its candidate SeVs, without requiring future or global states.% information.

	\section{Near-Optimal Number of Task Replicas} \label{sys_ana}
	In this section, we carry out performance analysis to derive the optimal number of task replicas, in terms of minimizing the average task offloading delay while satisfying the failure probability constraint.
	
	To enable the analysis, we consider a single-lane road system, where TaVs and SeVs are modeled by two independent one dimensional Poisson point processes (PPPs) with densities $\gamma_t$ and $\gamma_s$ (in vehicles per $\mathrm{km}$), respectively. We consider a homogeneous and stationary system where TaVs have the same task arrival rate $\lambda_0$, SeVs have the same service rate $\mu_c$, for $\forall n$, and the packet erasure probability is also identical, denoted by $p_e$, for $\forall t,n$. 
	The transmission rate for task upload is $r_u$, for $\forall t,n$. The result feedback is assumed to be successful with negligible delay, since the output data size $L_o$ is usually small \cite{you2016energy,sun2017emm}. 
	
	As discussed in the previous section, task upload delay is \emph{not} related to the number of task replicas due to multicast. Therefore, we do not focus on the task upload delay. Details on the analysis of transmission delay in vehicular networks can be found in \cite{Yao2013Delay, Yao2013Performance}. 
	Nevertheless, we remark that multicast may lead to packet collisions if multiple TaVs transmit data simultaneously. 
	A packet of the typical TaV may collides with two kinds of TaVs according to their relative locations.
	1) Collisions with other TaVs within the carrier sensing range of the typical TaV only occurs if multiple TaVs transmit at the same backoff slot.
	2) TaVs which are outside the carrier sensing rang of the typical TaV while within the communication range of candidate SeVs are called hidden TaVs. The typical TaV cannot sense whether or not hidden TaV is transmitting, and vice versa. In this case, collision occurs if the whole transmission periods of typical and hidden TaVs are overlapped. 
	Collisions mainly happen with hidden TaVs. However,  by choosing a proper communication range such that the carrier sensing range of each TaV is larger than $2R$, no hidden TaVs exist in the network. The collision probability is then negligible if the contention window size is large and the backoff slot length is short, which are usually true in the realistic VEC systems. 
	
	%A packet of the typical TaV mainly collides with hidden terminals which are outside the carrier sensing range. However, by choosing a proper communication range such that the carrier sensing range of each TaV is larger than $2R$, no hidden terminals exist in the network. Then a collision occurs only when multiple TaVs transmit at the same backoff slot. The collision probability is negligible if the contention window size is large and the backoff slot length is short, which are usually true in the realistic VEC systems. 
	
	%if a proper backoff contention window size is selected, and the carrier sensing range is larger than $2R$ (so that no hidden terminals exist in the network, and .)
	%The broadcast delay in vehicular networks has been widely investigated. For example, in \cite{Yao2013Delay, Yao2013Performance}, the backoff procedure is modeled by a discrete-time Markov chain, and the  
	
	\subsection{Average Task Arrival Rate at each SeV}
	To analyze the task execution delay and derive the optimal number of task replicas, we first characterize the average task arrival rate at each SeV.
	
	Consider a typical SeV $n_0$ within the communication range of the typical TaV. Denote the set of TaVs within the communication range of SeV $n_0$ by $\mathcal{Y}_0$. For any TaV $i\in \mathcal{Y}_0$, denote the number of its candidate SeVs by $Y_i$. Assume that TaV $i$ randomly selects $K$ candidate SeVs for task replication, if $Y_i\geq K$. If $K<Y_i$, the replicas of the task are offloaded to all the candidate SeVs. Then the probability that SeV $n_0$ is selected by TaV $i$ is $\frac{\min\{K, Y_i\}}{Y_i}$. Denote the average task arrival rate at SeV $n_0$ by $\lambda_c$, which can be written as
	\begin{align}
	\lambda_c&=\mathbb{E}\left[\sum_{i=1}^{|\mathcal{Y}_0|} \frac{\min\{K, Y_i\}}{Y_i} \lambda_0\right] . 
	\end{align} 
	
	Let $\bar{\gamma}_t=2R\gamma_t$, $\bar{\gamma}_s=2R\gamma_s$ be the average number of TaVs and SeVs within length $2R$, respectively.
	An upper bound of $	\lambda_c$ is given by the following Lemma.
	
	\begin{lemma} \label{arrival_rate}
		Given the number of task replicas $K$, an upper bound of the average task arrival rate is given by
		\begin{align}
		\lambda_c\leq \left(\bar{\gamma}_t+1\right) \lambda_0 K  \sum_{k=1}^{\infty} \frac{1}{k}\frac{\bar{\gamma}_s^k}{k!} e^{-\bar{\gamma}_s} .
		\end{align}
		%	where $\bar{\gamma}_t=2R\gamma_t$, $\bar{\gamma}_s=2R\gamma_s$.
	\end{lemma}
	\begin{proof}
		See Appendix \ref{a1}.
	\end{proof}

	\subsection{Task Execution Delay}
	Denote the number of candidate SeVs within the communication range of the typical TaV by $N_s$, which is a random variable following Poisson distribution with rate $\bar{\gamma}_s$.
	We only consider the case when $N_s\geq1$.
	Define $S$ as the number of SeVs that can successfully receive the task from the typical TaV.
	Since the packet erasure behavior of each SeV is independent of others, $S$ follows binomial distribution. Specifically, 
	\begin{align}
	\mathbb{P}[S=k]= \begin{cases}\binom{N_s}{k} (1-p_e)^k p_e^{N_s-k},&N_s<K,~k=1,2,\cdots,N_s, \\
	\binom{K}{k}(1-p_e)^k p_e^{K-k},&N_s\geq K, k=1,2,\cdots,K.\end{cases}   
	\end{align}
	
	%if $N_s<K$, $\mathbb{P}[S=k]=\binom{N_s}{k} (1-p_e)^k p_e^{N_s-k}, k=1,2,\cdots,N_s$; otherwise $\mathbb{P}[S=k]=\binom{K}{k}(1-p_e)^k p_e^{K-k}, k=1,2,\cdots,K$.
	
	At each SeV, task execution is modeled by an M/M/1 queue with arrival rate $\lambda_c$ and service rate $\mu_c$. Therefore, the task execution delay follows exponential distribution with mean $\frac{1}{\mu_c-\lambda_c}$.
	Since the result feedback is assumed to be with negligible delay without packet loss, 
	%the typical TaV can receive the result as soon as an SeV completes the execution.
	%Therefore, 
	the average task execution delay is the first order statistics of $S$ exponential distributions, which equals to $\frac{1}{S(\mu_c-\lambda_c)}$.
	Let $K_s\triangleq \min\{K,N_s\}$. 
	Given the number of candidate SeVs $N_s$, the expected task execution delay is $\sum_{k=1}^{K_s}\mathbb{P}[S=k]\frac{1}{S(\mu_c-\lambda_c)}$.
	Since $N_s$ is a random variable following Poisson distribution with rate $\bar{\gamma}_s$, the expected task execution delay, denoted by $D_c$, can be given by
	\begin{align}\label{exe_delay}
	D_c=\sum_{N_s=1}^{\infty} \frac{\bar{\gamma}_s^{N_s}}{N_s!} e^{-\bar{\gamma}_s} 
	\sum_{k=1}^{K_s} \binom{K_s}{k}(1-p_e)^k p_e^{K_s-k}\frac{1}{k}\frac{1}{\mu_c-\lambda_c},
	\end{align}
	where the number of task replicas $K$ is the optimization variable.
	
	%By minimizing $D_c$, we can obtain the optimal number of task replicas $K^*$, i.e., 
	%\begin{align}
	%	K^*=\argmin_{K=1,2,\cdots} D_c.
	%\end{align}
	
	An approximation to the optimal number of replicas that minimizes $D_c$ is given in the following Theorem.
	
	\begin{theorem} \label{theo1}
		The optimal number of task replicas that minimizes the average task execution delay $D_c$ is approximated by
		\begin{align}
		\tilde{K}^*= \frac{\mu_c}{2\lambda_0  \left(\bar{\gamma}_t+1\right) \left( \frac{1}{\bar{\gamma}_s}+ \frac{1}{\bar{\gamma}_s^2}\right)}.
		\end{align}
	\end{theorem}
	\begin{proof}
		See Appendix \ref{a2}.
	\end{proof}
	
	As shown in Theorem \ref{theo1}, the near-optimal number of task replicas $\tilde{K}^*$ related to four key parameters of the VEC system. Specifically, $\tilde{K}^*$ is proportional to the service capability $\mu_c$, inversely proportional to the task arrival rate $\lambda_0$, and approximately proportional to the SeV density $\bar{\gamma}_s$ and TaV density $\frac{1}{\bar{\gamma}_t}$. Remark that, as the number of task replicas is an integer in practice, we can round $\tilde{K}^*$ to its nearest integer for implementations.

	\subsection{Task Failure Probability}
	
	Define $P_f$ as the task failure probability.
	If $N_s<K$, $P_f=p_e^{N_s}$; otherwise $P_f=p_e^{K}$.
	Therefore, $P_f$ can be written as
	\begin{align}
	P_f=\sum_{N_s=1}^{K} \frac{\bar{\gamma}_s^{N_s}}{N_s!} e^{-\bar{\gamma}_s} p_e^{N_s}+\sum_{N_s=K+1}^{\infty} \frac{\bar{\gamma}_s^{N_s}}{N_s!} e^{-\bar{\gamma}_s}  p_e^{K}.	
	\end{align}
	
	\begin{lemma} \label{fail_prob}
		To guarantee the failure probability constraint $P_f\leq\theta_f$, a lower bound of the number of replicas is 
		\begin{align}
		K\geq \left\lceil\frac{\ln\theta_f }{\ln p_e} \right\rceil.
		\end{align}
	\end{lemma}
	\begin{proof}
		See Appendix \ref{a3}.
	\end{proof}
	
	Combining Theorem \ref{theo1} and Lemma \ref{fail_prob}, we obtain the approximation to the optimal number of replicas in the following Corollary.
	
	\begin{corollary} \label{corollary_1}
		To minimize the task execution delay while satisfying the task failure probability threshold, the number of task replicas should be set to
		\begin{align} \label{K_star}
		K^*=\max \left\{  \text{round}\left(\tilde{K}^*\right),   \left\lceil\frac{\ln\theta_f }{\ln p_e} \right\rceil   \right\}.
		\end{align}
	\end{corollary}
	
	%\begin{remark}
	%	In the VEC system, the number of task replicas $K^*$ should be decided in a centralized manner according to regional states. For example, RSUs can collect global state information including the service rate, task arrival rate, packet erasure probability as well as the densities of TaVs and SeVs periodically. Afterwards, the number of replicas can be optimized based on the average of the empirical state information, and multicast to the TaVs.
	%\end{remark}
	
	\section{Distributed Task Replication Algorithm: A Combinatorial Multi-Armed Bandit based Approach} \label{algo}
	Based on the optimized number of task replicas $K^*$, we design a distributed task replication algorithm in this section. Recall that the instantaneous global states, such as the number of other TaVs in the neighborhood, the workloads and channel environments of SeVs are very challenging to be acquired by the TaV. Accordingly, the TaV cannot know a priori which candidate SeV can provide the fastest computation for each task.
	
	To overcome the aforementioned challenge, we propose a solution called \emph{learning while offloading}: the TaV can observe the delay performance of its candidate SeVs while offloading tasks, and learn about which subset of SeVs should be selected to minimize the offloading delay. 
	
	We further assume that the TaV makes task replication decisions only when a task becomes the head of the queue. 
	On the one hand, the TaV may face different candidate SeVs for the following tasks, so that the offloading decisions made in advance may not be able to be implemented.
	On the other hand, making offloading decisions for multiple tasks simultaneously complicates the optimization problem, which might be solved by reinforcement learning technique, but with very high complexity. 
	
	Then the task replication problem is an online sequential decision making problem, which is very similar to the MAB problem. 
	In the classical MAB problem, a player faces a fixed number of base arms with unknown rewards, and pulls one at a time to learn the reward distributions while maximizing the cumulative rewards over time. The major challenge of the MAB problem is the \emph{exploration-exploitation tradeoff} during the learning process: to explore different arms and learn a more accurate reward distribution, or to exploit the current knowledge and choose the empirically optimal arm. Such problem has been widely investigated, and upper confidence bound (UCB) based algorithms have been  proposed with performance guarantee \cite{auer2002finite}.
	
	An extension of MAB is called CMAB, in which a super arm, composed of a subset of base arms, is selected at a time. The player observes the rewards of all the selected base arms, and obtains a reward from the super arm, which can be either a linear or non-linear function of the rewards of base arms \cite{chen2013cmab,chen2016cmab}.
	Our task replication framework resembles the CMAB framework: each candidate SeV is a base arm with an unknown delay (loss) distribution, and the TaV is the player who selects a subset of SeVs $\mathcal{S}_t$ for each task. Then the offloading delay of SeV $n\in \mathcal{S}_t$ is observed upon result feedback. Note that there might be a packet loss or very long delay. 
	We define $d_{\text{max}} $ as the maximum offloading delay that is allowed for each task replica. 
	Specifically, for the $t$-th task, the offloading delay of SeV $n\in \mathcal{S}_t$ is 
	\begin{align}
	d(t,n)=\min \{d_u(t,\mathcal{S}_t)+ d_c(t,n) +d_d(t,n),d_{\text{max}} \},
	\end{align}
	%$d(t,n)=\min \{d_u(t,\mathcal{S}_t)+ d_c(t,n) +d_d(t,n),d_{\text{max}} \} $, 
	and the offloading delay of the task is $d(t,\mathcal{S}_t)=\min_{n\in \mathcal{S}_t} d(t,n) $, which is a non-linear function of the individual offloading delay. Since the maximum delay $d_{\text{max}}$ can reflect packet loss, problem \textbf{P1} is transformed to \textbf{P2}:
	%\begin{align}
	%\textbf{P2:}~\min_{\mathcal{S}_1,...,\mathcal{S}_T} \frac{1}{T}\sum_{t=1}^{T} d(t,\mathcal{S}_t) .
	%\end{align}
	\begin{subequations}
		\begin{align}
		\textbf{P2:}~\min_{\mathcal{S}_1,...,\mathcal{S}_T} &\frac{1}{T}\sum_{t=1}^{T} d(t,\mathcal{S}_t) \\
		\text{s.t.}~~~ &\mathcal{S}_t\subseteq\mathcal{N}_t, ~t=1,\cdots,T.
		\end{align}
	\end{subequations}

	However, existing algorithms for CMAB problem cannot be implemented directly. In our problem, the candidate SeV set $\mathcal{N}_t$ changes accross time, with unknown appearance and disappearance time. Existing algorithms in \cite{chen2013cmab,chen2016cmab} should be revised in order to adapt to such a dynamic vehicular environment.
	
	As shown in Algorithm 1, we propose a learning-based task replication algorithm (LTRA). The offloading delay is first normalized according to 
	\begin{align} \label{d_nor}
	\tilde{d}(t,n)=\frac{d(t,n)}{d_{\text{max}}},
	\end{align}
	with $\tilde{d}(t,n)\in (0,1]$. 
	For any SeV $n$, denote the empirical probability density function (PDF) of $ 1-\tilde{d}(t,n)$ by $\hat D_n$, and the cumulative distribution function (CDF) by $\hat F_n$.
	Let $t_n$ indicate that the $n$-th SeV occurs upon offloading the $t_n$-th task. Let $k_{t,n}$ be the number of tasks offloaded to SeV $n$ among the first $t$ tasks, and $\alpha$ a constant factor.

	\begin{algorithm}[!t]
		\caption{Learning-based Task Replication Algorithm}
		\begin{algorithmic}[1]
			\For {$t=1,...,T$}
			\If { Any new SeV $n_s \in \mathcal{N}_t$, $n_s \notin \mathcal{N}_{t-1}$ appears}
			\State Connect to any subset $\mathcal{S}_t \subseteq \mathcal{N}_t$ once, with $n_s\in \mathcal{S}_t$ and $|\mathcal{S}_t|=K^*$.
			\State Wait for result and record delay $d(t,n)$, $\forall n\in \mathcal{S}_t$.
			\State Update empirical CDF $\hat F_n$ according to $\tilde{d}(t,n)$ and the selected times $k_{t,n}\leftarrow k_{t-1,n}+1$ for $\forall n\in \mathcal{S}_t$. 
			\Else
			\State Update CDF $\underline{F}_n$ according to \eqref{utility}, and calculate the corresponding PDF $\underline{D}_n$, for $\forall n\in \mathcal{N}_t$.
			%		\State For each $n\in\mathcal{N}_t$, define CDF $\underline{F}_n$ as	
			\State Select a subset of SeVs $\mathcal{S}_t$ according to \eqref{bestset}.		
			\State Offload the $t$-th task to SeV $\forall n \in \mathcal{S}_t$.
			\State Wait for result and record delay $d(t,n)$, $\forall n\in \mathcal{S}_t$.
			\State Update $\hat F_n$ and $k_{t,n}\leftarrow k_{t-1,n}+1$, $\forall n\in \mathcal{S}_t$.
			\EndIf
			\EndFor		
		\end{algorithmic}
	\end{algorithm}
	
	In Algorithm 1, Lines 2-5 are the initialization phase, which is called at the start of the learning process as well as the time when new candidate SeV occurs.
	The TaV selects a subset of $K^*$ SeVs that contains the newly appeared SeVs, where $K^*$ is obtained according to \eqref{K_star}. Note that $\mathcal{N}_{0}=\emptyset$, $k_{0,n}=0$, and if, occasionally, the newly appeared SeVs are more than $K^*$, we allow the TaV to offload the tasks to all the new SeVs. 
	
	Lines 6-11 are the main loop of LTRA. Taking into consideration the occurrence time $t_n$ of SeV $n$, a CDF $\underline{F}_n(x)$  is defined as
	\begin{align} \label{utility}
	\underline{F}_n(x) \!\! = \!\!    \begin{cases}\max \! \left\{ \hat{F}_n(x) \!\! - \!\!  \sqrt{\frac{\alpha\ln (t-t_n)}{k_{t-1,n}}},0 \right\},\!\! &\!\! 0\leq x <1, \\
	1,&\!\! x=1.\end{cases}        
	\end{align}	
	Let $\underline{D}_n$ be the distribution of $\underline{F}_n$, and $\underline{D}=\underline{D}_1 \times \underline{D}_2 \times ... \times \underline{D}_{|\mathcal{N}_t|} $ the joint distribution over all candidate SeVs. The subset of SeVs is selected according to
	\begin{align} \label{bestset}
	\mathcal{S}_t= \argmin_{\mathcal{S} \subseteq \mathcal{N}_t ,|\mathcal{S} |=\min\{|\mathcal{N}_t|, K^*\}} \mathbb{E}_{\underline{D}}\left[\min_{n \in \mathcal{S}}d(t,n)\right]. 
	\end{align}
	The calculation of $\mathcal{S}_t$ is a minimum element problem, which can be solved by greedy algorithms \cite{goel2010how}.
	Then the TaV multicasts the input data of the task to the selected SeVs $n \in \mathcal{S}_t$, waits for the results for a maximum time length $d_\text{max}$, and records the corresponding delay $d(t,n)$. Finally, the TaV updates the empirical CDF $\hat F_n$ according to normalized delay $\tilde{d}(t,n)$, as well as the selected times $k_{t,n}$.
	
	We remark that, the proposed LTRA learns the \emph{entire delay distribution} of candidate SeVs, and is able to balance the \emph{exploration-exploitation tradeoff} during the learning process. Due to the non-linearity of the loss function $d(t,\mathcal{S}_t)=\min_{n\in \mathcal{S}_t} d(t,n) $, the offloading decision $\mathcal{S}_t$ cannot be decided merely by the mean delay of candidate SeVs, but their joint distribution. Therefore, the TaV records the empirical CDF $\hat F_n$ while learning.
	Meanwhile, $\underline{F}_n$ is designed to guide the offloading decisions.
	For an SeV with fewer selected times $k_{t,n}$, the padding term $\sqrt{\frac{\alpha\ln (t-t_n)}{k_{t-1,n}}}$ is large, so that the TaV finds it a good choice to provide possible low delay performance and \emph{explores} it.
	The TaV also tends to \emph{exploit} SeVs with lower offloading delay according to the empirical CDF $\hat F_n$. 
	Furthermore, it is easy to see that for $0\leq x\leq 1$, $\underline{F}_n(x) <\hat{F}_n(x) $, i.e., $\underline{F}_n(x)$ \emph{first-order stochastically dominates} $\hat{F}_n(x) $.
	%By first-order stochastically dominating $\hat F_n$ during the learning process, 
	The CDF $\underline{F}_n$ provides more optimistic estimations to those SeVs with less information learned, to balance the tradeoff between exploration and exploitation during the learning process.
	
	\subsection{Performance Analysis}
	To characterize the performance of the proposed LTRA, we assume that the candidate SeV set does not change during the considered time period, i.e., $\mathcal{N}_t=\mathcal{N}$, for $\forall t$. Moreover, the delay distribution $d(t,n)$ is independently and identically distributed (i.i.d.) with respect to the task index $t$. In the simulation results, we will show that without these two assumptions, the proposed algorithm still works well.
	
	For the $t$-th task, let the delay vector of SeVs be $\d_t=(d(t,1),...,d(t,N))$, where $N= |\mathcal{N}|$. Define the loss function as $L(\d_t, \mathcal{S}_t)=\min_{n\in \mathcal{S}_t}d(t,n)$, and let $\mu_\mathcal{S}=\mathbb{E}[L(\d_t, \mathcal{S}_t)], \forall t$.
	Furthermore, let $\mathcal{S}^*=\argmin_{\mathcal{S} \subseteq \mathcal{N}, |\mathcal{S} |=\min\{N, K^*\} }\mu_\mathcal{S}$ denote the optimal subset of SeVs with minimum expectation of offloading delay, and $\mu_{\mathcal{S}^*}=\min_{\mathcal{S} \subseteq \mathcal{N}, |\mathcal{S} |=\min\{N, K^*\} }\mu_\mathcal{S}$.
	
	The performance metric to characterize the learning algorithm is called \emph{learning regret}, which is defined as
	\begin{align}
	R_T=\mathbb{E}\left[   \sum_{t=1}^{T}L(\d_t, \mathcal{S}_t)   \right]   -T\mu_{\mathcal{S}^*}.
	\end{align}
	The learning regret is the expectation of the performance loss caused by learning process, which is compared to the genie-aided case where the TaV knows the exact delay distributions of candidate SeVs.
	
	For any suboptimal subset of SeVs $\mathcal{S}\subseteq \mathcal{N}$ with $|\mathcal{S} |=\min\{N, K^*\} $, denote the expectation of the performance gap by $\Delta_\mathcal{S}=(\mu_\mathcal{S}-\mu_{\mathcal{S}^*})/d_{\text{max}}$.
	Let 
	\begin{align}
	\Delta_n \!\!= \!\! \min\left\{\Delta_\mathcal{S}| \mathcal{S}\!\subseteq \!\mathcal{N}, |\mathcal{S} |  \!  = \!\min\{\!N, K^* \!\}, n\in  \mathcal{S},  \mu_\mathcal{S}>\mu_{\mathcal{S}^*}    \right\}. \nonumber
	\end{align}
	
	In the following theorem, we show an upper bound of the learning regret of the proposed LTRA.
	
	\begin{theorem} \label{theo}
		Let $\alpha=\frac{2}{3}$, then $R_T$ is upper bounded by:
		\begin{align} \label{regret_bound}
		R_{T}\leq d_{max}\left(C_1 K \sum_{n\in\mathcal{N}}\frac{\ln T}{\Delta_n}+C_2 \right),
		\end{align}
		where $C_1=2136$ and $C_2=\left(\frac{\pi^2}{3}+1\right)N$ are two constants.
	\end{theorem}
	\begin{proof}
		See Appendix \ref{a4}.
	\end{proof}
	
	Theorem \ref{theo} indicates that, the learning regret of LTRA grows logarithmically with respect to the number of tasks $T$, and is also related to the performance gap $\Delta_n$ and the number of candidate SeVs $N$.

	\subsection{Implementation Considerations}
	In reality, the observed offloading delay $d(t,n)$ is continuous within range $(0,d_{\text{max}}]$. As the number of tasks $t$ grows, the proposed LTRA suffers from high storage cost to record all the offloading delay, as well as high computational complexity to calculate $\mathcal{S}_t$ according to \eqref{bestset}. These two phenomena violate the motivation for task offloading, i.e., the TaV has limited computing and storage resources.
	
	A feasible solution is to \emph{discretize} the empirical CDF $\hat F_n$. The discretization level is denoted by $l$, and the support of the discretized CDF is given by $\left\{0,\frac{1}{l},\frac{2}{l},...,\frac{l-1}{l}\right\}$, after partitioning range $(0,1]$ (the range of $\tilde{d}(t,n)$) into $l$ segments with equal intervals. If value $1-\tilde{d}(t,n)$ belongs to $\left[\frac{j}{l},\frac{j+1}{l}\right)$,  the empirical CDF $\tilde F_n$ is updated by value $\frac{j}{l}$.
	Discretization leads to additional learning regret, which can still be bounded according to \cite{chen2016cmab}.
	
	Another issue is that the $t$-th task may be offloaded before the TaV collects all the result feedbacks of the  previous $t-1$ tasks. In this case, a simple way is to use the up-to-date learned information to guide the offloading decisions. 
	
	\begin{remark}
		The relationship between Section \ref{sys_ana} and Section \ref{algo} is remarked here:
		The optimal number of task replicas $K^*$ can be provided in a large time-scale based on global conditions of a region. Meanwhile, LTRA works in a small time-scale, using $K^*$ as an input parameter, to guide the offloading decisions in a distributed manner.
	\end{remark}

	\section{Numerical and Simulation Results}\label{sim}
	
	In this section, we carry out simulations to validate the theoretical results and evaluate the proposed task replication algorithm. We first compare the approximation to the optimal number of task replicas obtained from Section \ref{sys_ana} with numerical and simulation results, and then simulate the proposed LTRA under a realistic traffic scenario.

	\subsection{Validation of System-Level Performance Analysis}
	
	Both numerical and simulation results are shown in this subsection, to validate the theoretical analysis in Section \ref{sys_ana}.
	Recall that the task upload delay is not related to the number of task replicas due to multicasting, thus we only focus on the task execution delay.
	
	The density of vehicles, including TaVs and SeVs, is set to $\gamma_t+\gamma_s=25$ vehicles per $\mathrm{km}$. The packet erasure probability is  $p_e=0.02$, and the communication range of each TaV is $R=200\mathrm{m}$. The service rate of each SeV is set to $\mu_c=10$. The theoretical task execution delay is calculated according to \eqref{exe_delay}, and the corresponding simulation result is obtained via Monte Carlo method, where a $10 \mathrm{km}$ single-lane road is considered with $10^6$ realizations. 
	
	\begin{figure}[!t]
		\centering	
		\subfigure[The density ratio of TaV to SeV is $\frac{\gamma_t}{\gamma_s}=\frac{1}{4}$.]{\label{sys_lambda_beta02}			
			\includegraphics[width=0.6\textwidth]{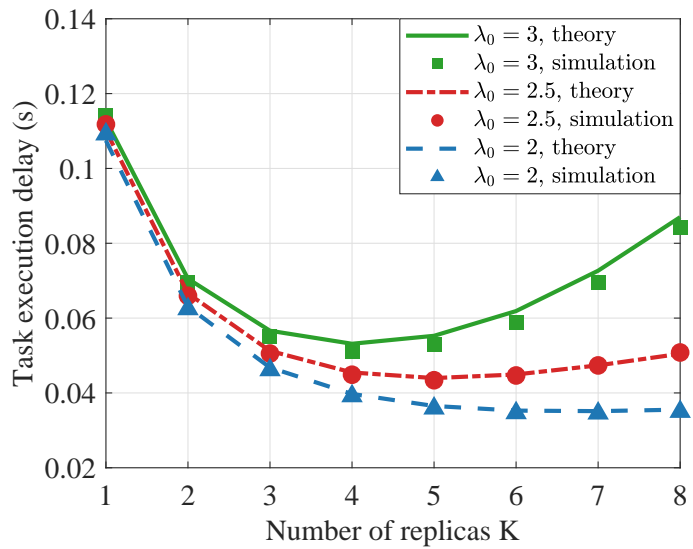}}
		%	\hspace{8mm}		
		\subfigure[The density ratio of TaV to SeV is $\frac{\gamma_t}{\gamma_s}=\frac{1}{3}$.]{\label{sys_lambda_beta025}	
			\includegraphics[width=0.6\textwidth]{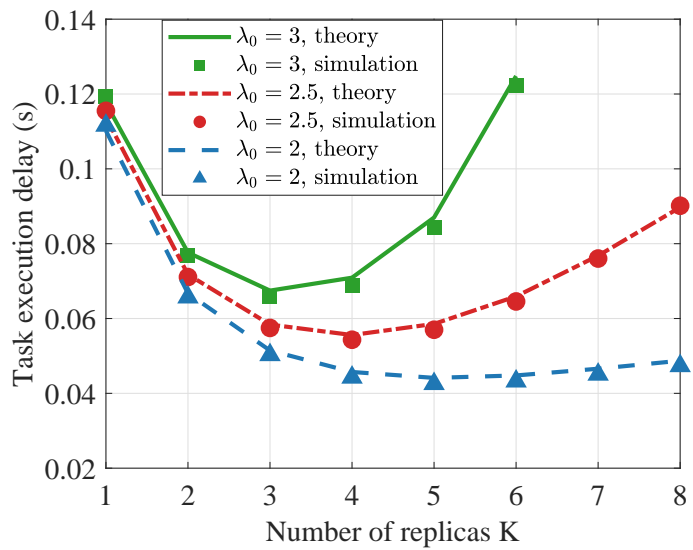}}
		\caption{Numerical and simulation results of the average task execution delay with respect to the number of task replicas under different task arrival rates $\lambda_0$.}
		\label{sys_lambda}
	\end{figure}
	
	\begin{figure}[!t]
		\centering	
		\includegraphics[width=0.6\textwidth]{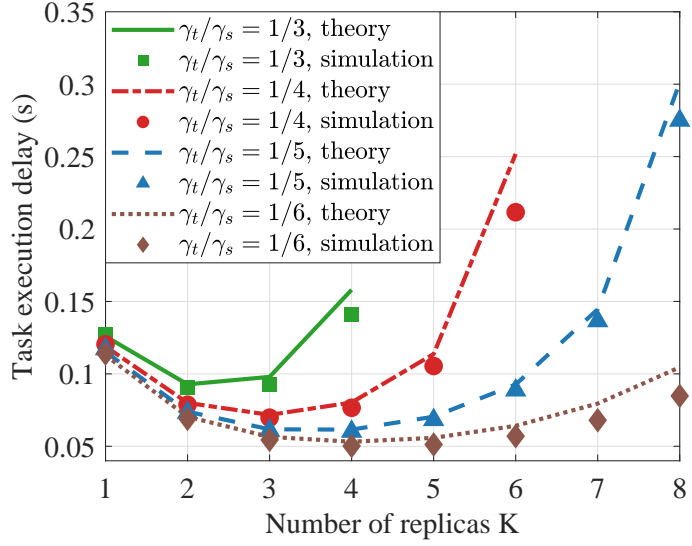}	
		\caption{Numerical and simulation results of the average task execution delay with respect to the number of task replicas under different TaV and SeV densities, with $\lambda_0=4$.}
		\label{sys_beta}
	\end{figure}

	Fig. \ref{sys_lambda} and Fig. \ref{sys_beta} show the average task execution delay with respect to the number of task replicas $K$ under different task arrival rates $\lambda_0$ and TaV to SeV density ratios $\frac{\gamma_t}{\gamma_s}$. Given $\lambda_0$ and $\frac{\gamma_t}{\gamma_s}$, the task execution delay first decreases with $K$ and then increases, and the optimal $K$ varies with the network conditions. A proper number of task replicas can improve the delay performance substantially, compared to the case without replications. For example, as shown in Fig. \ref{sys_lambda_beta025}, when $\lambda_0=4$ and $\frac{\gamma_t}{\gamma_s}=\frac{1}{3}$, task execution delay with $K=4$ replicas can be decreased by $50\%$ compared to that without replication ($K=1$).
	
	\begin{table*}[!t]
		\caption{Theoretical and simulation results of the optimal number of task replicas.}
		\label{opt_K}
		\centering
		\begin{tabular}{p{0.75cm}<{\centering}|p{0.75cm}<{\centering}|p{0.75cm}<{\centering}|p{0.75cm}<{\centering}|p{0.75cm}<{\centering}|p{0.75cm}<{\centering}||p{0.75cm}<{\centering}|p{0.75cm}<{\centering}|p{0.75cm}<{\centering}|p{0.75cm}<{\centering}|p{0.75cm}<{\centering}|p{0.75cm}<{\centering}}
			\hline
			$\lambda_0$ & $\gamma_t/\gamma_s$ & $K_{\text{theory}}^*$  & $K_{\text{sim}}^*$ & $\tilde{K}^*$ &$\tilde{K}_{\text{round}}^*$  &
			$\lambda_0$ & $\gamma_t/\gamma_s$ & $K_{\text{theory}}^*$  & $K_{\text{sim}}^*$ & $\tilde{K}^*$ &$\tilde{K}_{\text{round}}^*$\\		
			\hline
			\multirow{7}{*}{2} & 1                      &2&2&1.68&2 &
			2                     &\multirow{7}{*}{1}   &2&2&1.68&2 \\
			& 1/2                   &4&4&3.28&3 &
			2.5                 &                                    &1&1&1.34&1 \\	 
			& 1/3                   &5&5&4.65&5 &
			3                    &                                    &1&1&1.12&1 \\	    
			& 1/4                  &7&7&5.84&6 &
			3.5                 &                                    &1&1&0.96&1 \\	          
			& 1/5                   &8&8&6.89&7 &
			4                    &                                    &1&1&0.84&1 \\	    
			& 1/6                   &8&8&7.81&8 &
			4.5                 &                                    &1&1&0.74&1 \\	
			& 1/7                   &8&8&8.62&9 &
			5                    &                                    &1&1&0.67&1 \\	
			\hline     
			\multirow{7}{*}{3} & 1                      &1&1&1.12&1 &
			2                     &\multirow{7}{*}{1/3}&5&5&4.65&5 \\
			& 1/2                    &2&2&2.19&2 &
			2.5                 &                                    &4&4&3.72&4 \\	 
			& 1/3                   &3&3&3.10&3 &
			3                    &                                    &3&3&3.10&3 \\	    
			& 1/4                 &4&4&3.89&4 &
			3.5                 &                                    &3&3&2.66&3 \\	          
			& 1/5                   &5&5&4.59&5 &
			4                    &                                    &2&2&2.33&2 \\	    
			& 1/6                   &6&6&5.20&5 &
			4.5                 &                                    &2&2&2.07&2 \\	
			& 1/7                   &6&6&5.75&6 &
			5                    &                                    &2&2&1.86&2 \\	
			\hline  
			\multirow{7}{*}{4} & 1                      &1&1&0.84&1 &
			2                     &\multirow{7}{*}{1/4}&7&7&5.84&6 \\
			&1/2                    &2&2&1.64&2 &
			2.5                 &                                    &5&5&4.68&5 \\	 
			& 1/3                   &2&2&2.33&2 &
			3                    &                                    &4&4&3.89&4 \\	    
			& 1/4                 &3&3&2.92&3 &
			3.5                 &                                    &3&3&3.34&3 \\	          
			& 1/5                   &4&4&3.44&3 &
			4                    &                                    &3&3&2.92&3 \\	    
			& 1/6                   &4&4&3.90&4 &
			4.5                 &                                    &3&3&2.60&3 \\	
			& 1/7                   &4&4&4.31&4 &
			5                    &                                    &2&2&2.34&2 \\	
			\hline        
		\end{tabular}
		%     	\vspace{-0.1in}
	\end{table*}

	%The curves in Fig. \ref{sys_lambda_beta02} or Fig. \ref{sys_lambda_beta025} show that increasing the task arrival rate of each TaV leads to higher task execution delay and lower optimal $K$. The curves in Fig. \ref{sys_beta} or the comparison of Fig. \ref{sys_lambda_beta02} and Fig. \ref{sys_lambda_beta025} show that 

	Table \ref{opt_K} compares the optimal number of task replicas obtained from theory and simulation. Specifically, the task execution delay is calculated according to \eqref{exe_delay} from $K=1$ to $8$, and $K_{\text{theory}}^*$ is the corresponding $K$ that minimizes \eqref{exe_delay}, which is the optimal theoretical result. $K_{\text{sim}}^*$ is obtained via Monte Carlo simulation. $\tilde{K}^*$ is calculated according to Theorem \ref{theo1}, and $\tilde{K}_{\text{round}}^*=\text{round}\left(\tilde{K}^*\right)$	is the integer nearest to $\tilde{K}^*$, which is our approximated result. Remark that, the main contribution of the analysis is to derive the near-optimal number of task replicas $\tilde{K}^*$ that minimizes the average task execution delay, as shown in Theorem \ref{theo1}. Therefore, we mainly validate the accuracy of $\tilde{K}^*$ in this part.
	
	We can see from Table \ref{opt_K} that under most cases, our near-optimal solution $\tilde{K}_{\text{round}}^*$ is exactly the same as the optimal theoretical and simulation results. Occasionally, $\tilde{K}_{\text{round}}^*$ is not the optimal solution, but it is quite close-to-optimal, with a maximum difference of 1. We remark that, the number of task replicas is always an integer, thus a difference of 1 is a very small gap. Moreover, the task execution delay achieved by $K=\tilde{K}_{\text{round}}^*$ and $K=K_{\text{theory}}^*$ are very close even if $\tilde{K}_{\text{round}}^*\neq K_{\text{theory}}^*$. For example, in Table \ref{opt_K}, when $\lambda_0=2$ and $\gamma_t/\gamma_s=\frac{1}{4}$, $K_{\text{theory}}^*=7$ and $\tilde{K}_{\text{round}}^*=6$. According to Fig. \ref{sys_lambda_beta02}, the task execution delay at $K=7$ and $K=6$ are almost the same.
	In brief, the approximation given in Theorem \ref{theo1} provides an accurate estimate of the optimal number of task replicas, which can guide the efficient task replication from the system point of view.

	\subsection{Evaluation of the Proposed Algorithm under a Realistic Traffic Scenario}
	
	\begin{figure*}[!t]
		\centering	
		\subfigure[The density ratio of TaV to SeV is $\frac{\gamma_t}{\gamma_s}=0$.]{\label{algo_conv_beta0}			
			\includegraphics[width=0.48\textwidth]{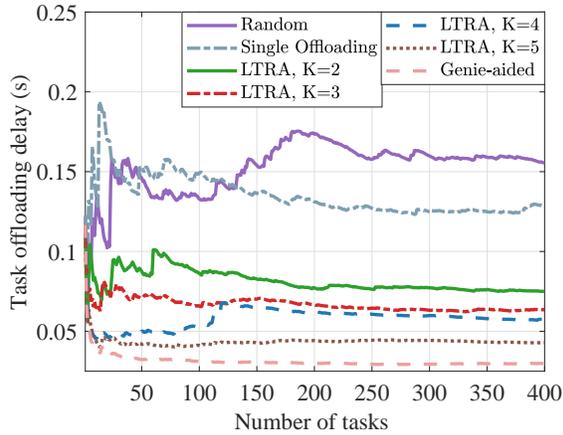}}		
		%	\hspace{8mm}
		\subfigure[The density ratio of TaV to SeV is $\frac{\gamma_t}{\gamma_s}=\frac{1}{7}$.]{\label{algo_conv_beta0125}	
			\includegraphics[width=0.48\textwidth]{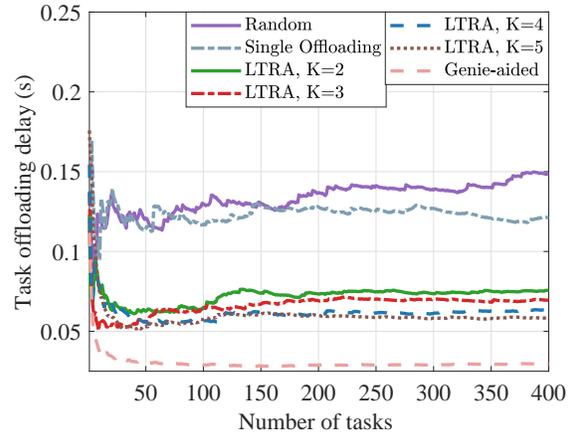}}
		\subfigure[The density ratio of TaV to SeV is $\frac{\gamma_t}{\gamma_s}=\frac{1}{4}$.]{\label{algo_conv_beta02}			
			\includegraphics[width=0.48\textwidth]{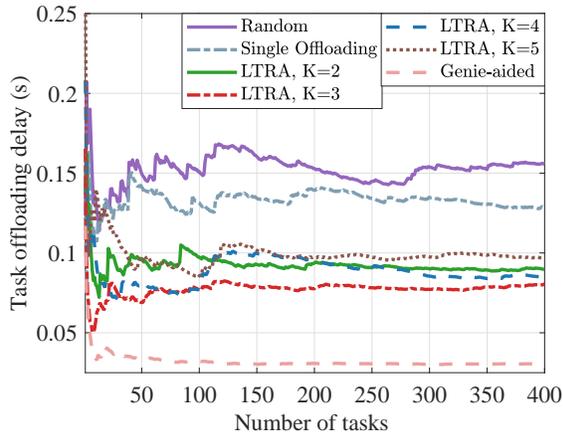}}	
		%	\hspace{8mm}
		\subfigure[The density ratio of TaV to SeV is $\frac{\gamma_t}{\gamma_s}=\frac{1}{3}$.]{\label{algo_conv_beta025}			
			\includegraphics[width=0.48\textwidth]{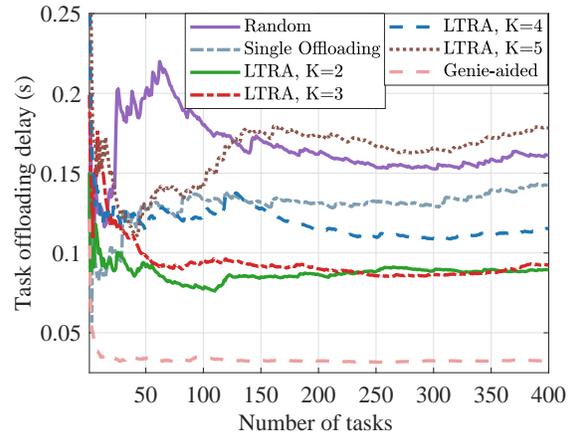}}	
		\caption{Average offloading delay of the proposed LTRA under the realistic traffic scenario, with $\lambda_0=4$.}
		\label{algo_conv}
	\end{figure*}
	
	\begin{figure*}[!htb]
		\centering	
		\subfigure[$\frac{\gamma_t}{\gamma_s}=\frac{1}{4}$, average offloading delay.]{\label{algo_delay_02}			
			\includegraphics[width=0.48\textwidth]{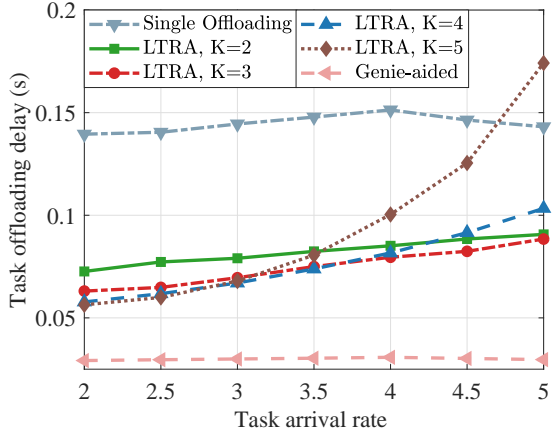}}		
		%	\hspace{8mm}
		\subfigure[$\frac{\gamma_t}{\gamma_s}=\frac{1}{4}$, task completion ratio.]{\label{algo_ddl_02}	
			\includegraphics[width=0.48\textwidth]{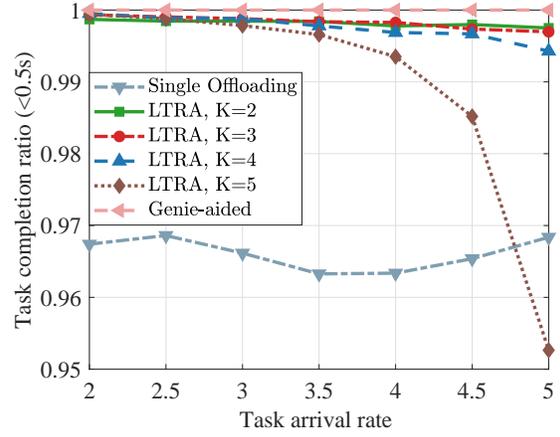}}
		\subfigure[$\frac{\gamma_t}{\gamma_s}=\frac{1}{3}$, average offloading delay]{\label{algo_delay_025}			
			\includegraphics[width=0.48\textwidth]{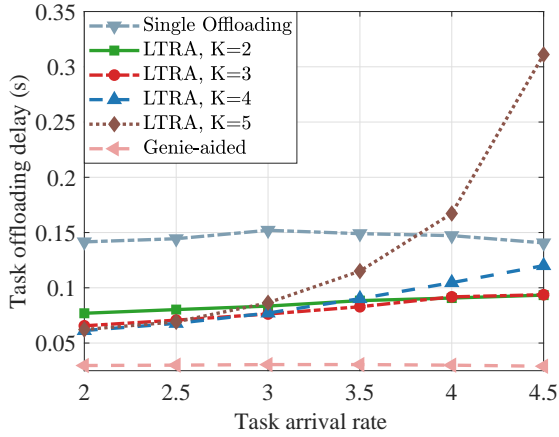}}	
		%	\hspace{8mm}
		\subfigure[$\frac{\gamma_t}{\gamma_s}=\frac{1}{3}$, task completion ratio.]{\label{algo_ddl_025}			
			\includegraphics[width=0.48\textwidth]{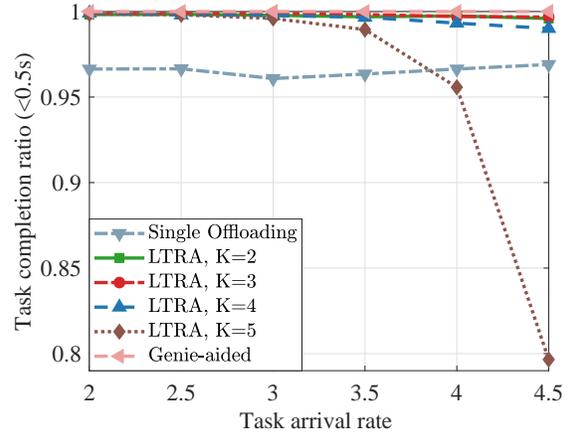}}	
		\caption{Average offloading delay and task completion ratio of the proposed LTRA under the realistic traffic scenario, with respect to task arrival rate $\lambda_0$.}
		\label{algo_delay_ddl}
	\end{figure*}
	
	To evaluate the offloading delay and reliability of the proposed task replication algorithm, we simulate a realistic traffic scenario in SUMO\footnote{http://www.sumo.dlr.de/userdoc/SUMO.html} using a $12\mathrm{km}$ segment of G6 Highway in Beijing, which is downloaded from Open Street Map (OSM)\footnote{http://www.openstreetmap.org/}. The traffic scenario is the same as that in Section VI of \cite{Sun2019TVT}. 
	The total arrival rate of vehicles in SUMO is set to $0.5$. According to the traffic flow information, the average number of vehicles per $\mathrm{km}$ is roughly $25$, which is close to the settings in the previous subsection.
	The maximum speed allowed of each TaV or SeV is $20\mathrm{m/s}$. 
	The output of SUMO is the car data, including the location, speed, angle of each vehicle at each time, which is imported to MATLAB for further simulations.
	
	For each task, the input data size is $L_i=1\mathrm{Mbits}$, and the output data size is considered to be negligible. For task upload, the path loss exponent is set to $2$, the channel bandwidth $W_u=10\mathrm{MHz}$, transmission power $P=0.5\mathrm{W}$, and noise power $N_0=10^{-13}\mathrm{W}$. The service rate $\mu_c(n)$ of each candidate SeV is uniformly distributed within $[8,12]$ tasks per second, and the packet erasure probability $p_e(t,n)$ is uniformly selected within $[0.01,0.03]$, so that the mean service rate $\mu_c=10$ and the mean packet erasure probability $p_e=0.02$ are the same as above. Moreover, parameter $\alpha$ in \eqref{utility} is set to $0.5$, the maximum delay $d_{\text{max}}=0.5\mathrm{s}$ and the discretization level $l=100$.
	
	The proposed algorithm is compared to: 1) \emph{Genie-aided policy}, where the TaV knows the exact global state information of all candidate SeVs, and offloads a single task to the SeV that can provide the minimum delay. Note that genie-aided policy cannot be realized in reality, which is used as a lower bound. 2) \emph{Random policy}, where the TaV randomly selects a single SeV for each task. 3) \emph{Single offloading policy}, which is also an MAB-based learning algorithm proposed in \cite{Sun2019TVT}, where each TaV selects a single SeV to offload each task, and learns the delay performance while offloading.
	
	Fig. \ref{algo_conv} shows the evolution of the average offloading delay with respect to the number of offloaded tasks, under different density ratios of TaV to SeV with task arrival rate $\lambda_0=4$. In Fig. \ref{algo_conv_beta0}, $\frac{\gamma_t}{\gamma_s}=0$ indicates that the target TaV can monopolize the SeVs instead of sharing with other TaVs. In this case, replication can fully exploit the diversity gain, and the more task replicas are offloaded, the lower task offloading delay can be achieved. When $K=5$, the average offloading delay is very close to that achieved by the genie-aided policy.
	As $\frac{\gamma_t}{\gamma_s}$ grows higher, i.e., more TaVs share the wireless channels and SeV computation resources, fewer number of task replicas may achieve better delay performance. For example, in Fig. \ref{algo_conv_beta025}, $K=5$ leads to serious overload, so that the delay performance is even worse than that of the random benchmark, while $K=2$ is the optimal choice. 
	Furthermore, the optimal number of task replicas obtained under the realistic traffic scenario coincides with the analytical results. 
	For example, in Fig. \ref{algo_conv_beta02}, $K=3$ is the optimal choice, which is the same as $\tilde{K}_{\text{round}}^*$  according to Table \ref{opt_K} when $\frac{\gamma_t}{\gamma_s}=\frac{1}{4}$ and $\lambda_0=4$.

	Fig. \ref{algo_delay_ddl} further plots the average offloading delay and task completion ratio under different task arrival rates and TaV to SeV density ratios. The deadline of each task is set to $d_{\text{max}}=0.5 \mathrm{s}$. Overall, task replication significantly improves the delay performance and task completion ratio compared with single task offloading.
	According to Fig. \ref{algo_delay_02} and Fig. \ref{algo_delay_025}, with a proper replication, the average offloading delay can be reduced by over $56\%$ when the task arrival rate is low ($\lambda_0=2$), and by $33\%$ when the task arrival rate is high ($\lambda_0=4.5$).
	Meanwhile, the task completion ratio can be improved from $97\%$ to over $99.6\%$ with proper number of task replicas, as shown in Fig. \ref{algo_ddl_02} and Fig. \ref{algo_ddl_025}.
	%We can also see that too many replicas significantly decreases the delay and reliability performance in the overload
	%We can also see that when $K=5$, the delay and reliability performance is extremely good when the the arrival rate is low, and is extremely bad 

	\section{Conclusions} \label{con}
	In this paper, we have investigated the task replication problem for delay minimization in the VEC system, and proposed a two-step solution to obtain realtime computing services.
	Given the general network conditions, we have approximated the optimal number of task replicas in closed-form, which is mainly related to densities of TaVs and SeVs, task arrival rates of TaVs, service capabilities of SeVs and packet erasure probability.
	Based on the analytical result, we have further designed LTRA based on CMAB theory, to enable distributed task replication in the highly dynamic vehicular environment.
	A realistic traffic scenario has been generated to evaluate the proposed task replication algorithm. Simulation results have shown that appropriate amount of task replications can improve the delay performance and task completion ratio significantly.
	Compared with single task offloading, task replication can reduce the average offloading delay by at least $30\%$, while improving the task completion ratio from $97\%$ to over $99.6\%$.
	
	Future research directions include to consider coded computation techniques \cite{Sun2019CMag,lee2018speeding,hierarchical,sun2019heterogeneous} and task cancellation principles \cite{gardner2017redundancy,joshi2017efficient}, to further improve the efficiency of resource utilization, while guaranteeing the quality of service of computation tasks.

	\appendices{}

	\section{Proof of Lemma 1} \label{a1}
	Since the number of TaVs $|\mathcal{Y}_0|$ is independent of the number of candidate SeVs $Y_i$, according to Wald's equation,
	\begin{align} \label{a1_1}
	\lambda_c&=\mathbb{E}\left[\sum_{i=1}^{|\mathcal{Y}_0|} \frac{\min\{K, Y_i\}}{Y_i} \lambda_0\right] 
	\nonumber\\
	&=\mathbb{E}[|\mathcal{Y}_0|]\mathbb{E}\left[ \frac{\min\{K, Y_1\}}{Y_1}  \right] \lambda_0
	\end{align}
	Since we are considering an SeV within the communication range of the typical TaV, it is equivalent to the case where the typical TaV is added to a PPP-distributed TaV set. According to Slivnyak's theorem \cite{stochastic_geo}, 
	\begin{align} \label{a1_2}
	\mathbb{E}[|\mathcal{Y}_0|]=\bar{\gamma}_t+1,
	\end{align}
	where $\bar{\gamma}_t=2R\gamma_t$ is the average number of TaVs within length $2R$.
	
	Let $\bar{\gamma}_s=2R\gamma_s$. The number of candidate SeVs around a TaV follows Poisson distribution with rate $\bar{\gamma}_s$.
	Therefore, 
	\begin{align} \label{a1_3}
	\mathbb{E}\left[ \frac{\min\{K, Y_1\}}{Y_1}  \right]&=\sum_{k=1}^{K} \frac{\bar{\gamma}_s^k}{k!} e^{-\bar{\gamma}_s}  + \sum_{k=K+1}^{\infty} \frac{K}{k}\frac{\bar{\gamma}_s^k}{k!} e^{-\bar{\gamma}_s}\nonumber\\ 
	&\leq K  \sum_{k=1}^{\infty} \frac{1}{k}\frac{\bar{\gamma}_s^k}{k!} e^{-\bar{\gamma}_s}.
	\end{align}
	
	Combining \eqref{a1_1}-\eqref{a1_3}, we can obtain
	\begin{align}
	\lambda_c&=\mathbb{E}[|\mathcal{Y}_0|]\mathbb{E}\left[ \frac{\min\{K, Y_1\}}{Y_1}  \right] \lambda_0\nonumber\\
	&\leq \left(\bar{\gamma}_t+1\right) \lambda_0 K  \sum_{k=1}^{\infty} \frac{1}{k}\frac{\bar{\gamma}_s^k}{k!} e^{-\bar{\gamma}_s} .
	\end{align}

	Therefore, Lemma \ref{arrival_rate} is proved.

	\section{Proof of Theorem 1} \label{a2}
	A conservative estimation of the average task execution delay, denoted by $\hat{D}_c$, is obtained by substituting the average task arrival rate $\lambda_c$ with its upper bound $\hat{\lambda}_c\triangleq\left(\bar{\gamma}_t+1\right) \lambda_0 K  \sum_{k=1}^{\infty} \frac{1}{k}\frac{\bar{\gamma}_s^k}{k!} e^{-\bar{\gamma}_s} $:
	\begin{align}
	\hat{D}_c=\sum_{N_s=1}^{\infty} \frac{\bar{\gamma}_s^{N_s}}{N_s!} e^{-\bar{\gamma}_s} 
	\sum_{k=1}^{K_s} {K_s \choose k}(1-p_e)^k p_e^{K_s-k}\frac{1}{k}\frac{1}{\mu_c-\hat{\lambda}_c}. \nonumber
	\end{align}
	To make the derivation tractable, we minimize $\hat{D}_c$ in the following, which upper bounds the expected task execution delay $D_c$. Accordingly, the optimal number of task replicas obtained can guarantee the stability of the system, i.e., if $\mu_c-\hat{\lambda}_c>0$, then $\mu_c-{\lambda}_c>0$.
	
	Let $\hat{c}=\left(\bar{\gamma}_t+1\right) \lambda_0 \sum_{k=1}^{\infty} \frac{1}{k}\frac{\bar{\gamma}_s^k}{k!} e^{-\bar{\gamma}_s}$, and thus $ \hat{\lambda}_c =\hat{c}K$. Recall that $K_s=\min\{N_s,K\}$. We have
	
	\begin{align}
	\hat{D}_c&=\sum_{N_s=1}^{\infty} \frac{\bar{\gamma}_s^{N_s}}{N_s!} e^{-\bar{\gamma}_s} 
	\sum_{k=1}^{K_s} {K_s \choose k}(1-p_e)^k p_e^{K_s-k}\frac{1}{k}\frac{1}{\mu_c-\hat{c}K} \nonumber\\
	&	=\sum_{N_s=1}^{K} \frac{\bar{\gamma}_s^{N_s}}{N_s!} e^{-\bar{\gamma}_s} 
	\sum_{k=1}^{N_s} {N_s \choose k}(1-p_e)^k p_e^{N_s-k}\frac{1}{k}\frac{1}{\mu_c-\hat{c}K} \nonumber\\
	&~~~~~~~~~~+\sum_{N_s=K+1}^{\infty} \frac{\bar{\gamma}_s^{N_s}}{N_s!} e^{-\bar{\gamma}_s} 
	\sum_{k=1}^{K} {K \choose k}(1-p_e)^k p_e^{K-k}\frac{1}{k}\frac{1}{\mu_c-\hat{c}K} \nonumber\\
	%\end{align}
	%\begin{align}	
	%	\hat{D}_c
	&\overset{(a)}{\approx} \frac{1}{\mu_c-\hat{c}K} \sum_{k=1}^{K} {K \choose k}(1-p_e)^k p_e^{K-k}\frac{1}{k}  \nonumber\\
	&=  \frac{1}{\mu_c-\hat{c}K} \sum_{k=1}^{K} \frac{K!}{kk!(K-k)!}(1-p_e)^k p_e^{K-k} \nonumber\\
	&=\frac{1}{\mu_s-\hat{c}K}  \sum_{k=1}^{K} \left[\frac{(K+1)!}{(k+1)!((K+1)-(k+1))!} \right.\nonumber\\  
	&~~~~~~~~~~\left.\frac{k+1}{k(K+1)(1-p_e)} (1-p_e)^{k+1} p_e^{(K+1)-(k+1)} \right] \nonumber\\
	%	&\overset{(b)}{\geq}\frac{\sum_{k=1}^{K}{K+1 \choose k+1}(1-p_e)^{k+1} p_e^{(K+1)-(k+1)}  }{K(\mu_s-\hat{c}K)(1-p_e)}   \nonumber\\
	&\overset{(b)}{\geq}\frac{ 1 }{K(\mu_s-\hat{c}K)(1-p_e)}  \sum_{k=1}^{K}{K+1 \choose k+1}(1-p_e)^{k+1} p_e^{(K+1)-(k+1)} \nonumber\\
	&\overset{(c)}{=}\frac{1- p_e^{K+1} -(K+1)(1-p_e)p_e^K  }{K(\mu_s-\hat{c}K)(1-p_e)}  \nonumber\\
	&\overset{(d)}{>}\frac{1 }{K(\mu_s-\hat{c}K)(1-p_e)}. \label{D_c_lower}
	\end{align}
	Approximation (a) is obtained by making $K$ replicas when $N_s<K$. (b) holds since $\frac{k+1}{k}\geq \frac{K+1}{K}$, for $\forall k=1,\cdots,K$.
	(c) follows the binomial expansion, and (d) is approximated due to the fact that packet erasure probability $p_e$ is close to $0$.
	
	By minimizing \eqref{D_c_lower}, we get an approximated number of task replicas that minimizes the average task execution time
	\begin{align} \label{approx_K}
	\tilde{K}^*=\frac{\mu_c}{2\hat{c}}.
	\end{align}
	
	Finally, we provide an approximation of $\hat{c}$. Observe that
	\begin{align}
	\frac{1}{k}=\frac{1}{k+1}+\frac{1}{(k+1)(k+2)}+O\left(\frac{1}{k^3}\right).
	\end{align}
	We have
	\begin{align}
	\sum_{k=1}^{\infty} \frac{1}{k}\frac{\bar{\gamma}_s^k}{k!} e^{-\bar{\gamma}_s} 
	&\approx\sum_{k=1}^{\infty}\frac{\bar{\gamma}_s^k}{(k+1)!} e^{-\bar{\gamma}_s}+\sum_{k=1}^{\infty}\frac{\bar{\gamma}_s^k}{(k+2)!} e^{-\bar{\gamma}_s} \nonumber\\
	&= \frac{1}{\bar{\gamma}_s} \sum_{k=1}^{\infty}\frac{\bar{\gamma}_s^{k+1}}{(k+1)!} e^{-\bar{\gamma}_s}+\frac{1}{\bar{\gamma}_s^2}\sum_{k=1}^{\infty}\frac{\bar{\gamma}_s^{k+2}}{(k+2)!} e^{-\bar{\gamma}_s}    \nonumber\\
	&=\frac{1}{\bar{\gamma}_s}\sum_{k=2}^{\infty}\frac{\bar{\gamma}_s^{k}}{k!} e^{-\bar{\gamma}_s}+\frac{1}{\bar{\gamma}_s^2}\sum_{k=3}^{\infty}\frac{\bar{\gamma}_s^{k}}{k!} e^{-\bar{\gamma}_s}\nonumber\\
	&=\frac{1-e^{-\bar{\gamma}_s}(1+\bar{\gamma}_s)}{\bar{\gamma}_s}+\frac{1-e^{-\bar{\gamma}_s} \left(1+\bar{\gamma}_s+\frac{\bar{\gamma}_s^2}{2}\right)}{\bar{\gamma}_s^2}\nonumber\\
	&\approx  \frac{1}{\bar{\gamma}_s}+ \frac{1}{\bar{\gamma}_s^2}.
	\end{align}
	
	Therefore, 
	\begin{align} \label{hat_c}
	\hat{c}=\left(\bar{\gamma}_t+1\right) \lambda_0 \sum_{k=1}^{\infty} \frac{1}{k}\frac{\bar{\gamma}_s^k}{k!} e^{-\bar{\gamma}_s}
	\approx\lambda_0 \left(\bar{\gamma}_t+1\right) \left( \frac{1}{\bar{\gamma}_s}+ \frac{1}{\bar{\gamma}_s^2} \right).
	\end{align}
	
	Substituting \eqref{hat_c} into \eqref{approx_K}, Theorem \ref{theo1} is proved.
	
	\section{Proof of Lemma 2} \label{a3}
	The lower bound of $P_f$ is given by $p_e^K$.
	Let $p_e^K\leq\theta_f$, we obtain a lower bound of the number of task replicas:
	\begin{align}
	K\geq \frac{\ln\theta_f }{\ln p_e}.
	\end{align}
	Since $K$ is an integer, we get Lemma \ref{fail_prob}.

	\section{Proof of Theorem 2} \label{a4}
	We prove that our delay minimization problem and the proposed LTRA is equivalent to the reward maximization problem and the stochastically dominant confidence bound (SDCB) algorithm considered in \cite{chen2016cmab}.
	
	The delay minimization problem \textbf{P2} can be transformed to a reward maximization problem:
	\begin{align} 
	&\min_{\mathcal{S}_1,...,\mathcal{S}_T} \frac{1}{T}\sum_{t=1}^{T}\min_{n \in \mathcal{S}_t}d(t,n)\nonumber\\
	=&d_{max}\min_{\mathcal{S}_1,...,\mathcal{S}_T} \frac{1}{T}\sum_{t=1}^{T}\min_{n \in \mathcal{S}_t}\tilde{d}(t,n)\nonumber\\
	\Leftrightarrow & \max_{\mathcal{S}_1,...,\mathcal{S}_T} \frac{1}{T}\sum_{t=1}^{T}\left[\max_{n \in \mathcal{S}_t}\left(1-\tilde{d}(t,n)\right)\right],
	\end{align}
	where $1-\tilde{d}(t,n)$ is considered as the reward of a base arm in \cite{chen2016cmab}, and the reward function $R(\d_t,\mathcal{S}_t)=\max_{n \in \mathcal{S}_t}\left(1-\tilde{d}(t,n)\right)\in [0,1]$. It is easy to see that the reward function is monotone with upper bound 1. Therefore, our model satisfies assumptions 1-3 in \cite{chen2016cmab}.
	
	In our proposed task replication algorithm, $\hat{F}_n$ records the empirical CDF of $1-\tilde{d}(t,n)$, and the CDF $\underline{F}_n$ is equivalent to that in the SDCB algorithm if $\mathcal{N}_t=\mathcal{N}$, for $\forall t$.
	Moreover, 
	\begin{align} 
	\mathcal{S}_t&= \argmin_{\mathcal{S} \subseteq \mathcal{N}_t ,|\mathcal{S} |=\min\{|\mathcal{N}_t|, K^*\}} \mathbb{E}_{\underline{D}}\left[\min_{n \in \mathcal{S}}d(t,n)\right] \nonumber\\
	&=\argmax_{\mathcal{S} \subseteq \mathcal{N}_t ,|\mathcal{S} |=\min\{|\mathcal{N}_t|, K^*\}}  \mathbb{E}_{\underline{D}}\left[\max_{n \in \mathcal{S}}\left(1-\tilde{d}(t,n)\right)\right] \nonumber
	\end{align}
	
	Therefore, the proposed LTRA is equivalent to SDCB algorithm when $\mathcal{N}_t=\mathcal{N}$, for $\forall t$.
	The performance bound is obtained directly from Theorem 1 in \cite{chen2016cmab}.

	%The learning regret $R_T$ is also equivalent to the definition in \cite{chen2016cmab}. Specifically,
	%\begin{align}
	%R_T&=\mathbb{E}\left[   \sum_{t=1}^{T}L(\d_t, \mathcal{S}_t)   \right]   -T\mu_{\mathcal{S}^*}\nonumber\\
	%&=d_{\text{max}}\left(T(1-\frac{\mu_{\mathcal{S}^*}}{d_{\text{max}}})- \mathbb{E}\left[   \sum_{t=1}^{T}\min_{n\in \mathcal{S}_t}\tilde{d}(t,n) \right]\right)
	%\end{align}

\end{document}